\documentclass{sig-alternate}
\makeatletter
\def\ps@headings{%
\def\@oddhead{\mbox{}\scriptsize\rightmark \hfil \thepage}%
\def\@evenhead{\scriptsize\thepage \hfil \leftmark\mbox{}}%
\def\@oddfoot{}%
\def\@evenfoot{}}
\makeatother
\usepackage{fixltx2e} 
\usepackage{amssymb,amsmath}
\usepackage{amsxtra}
\usepackage{amsfonts}
\usepackage{graphicx}
\usepackage[table]{xcolor}
\usepackage{framed}
\colorlet{shadecolor}{gray!50}

\newtheorem{theorem}{Theorem}[section]
\newtheorem{lemma}[theorem]{Lemma}
\newtheorem{proposition}[theorem]{Proposition}

\newtheorem{claim}[theorem]{Claim}

\def\E{\mathcal{E}}
\providecommand{\abs}[1]{\lvert#1\rvert}

\usepackage{color}
\newcommand{\todo}[1]{{\begin{small}\sffamily \color{green}TODO:  #1 \end{small}}}
\newcommand{\cancel}[1]{}

\makeatletter
\let\@copyrightspace\relax
\makeatother

\begin{document}

\title{From Caesar to Twitter \\ An Axiomatic Approach to Elites of Social Networks}
\numberofauthors{1}
\author{
\alignauthor
Chen Avin$^1$, Zvi Lotker$^1$, Yvonne-Anne Pignolet$^{2}$, Itzik Turkel$^1$\\
       \affaddr{$^1$Ben Gurion University of the Negev, Be'er-Sheva, Israel}\\
       \affaddr{$^2$ ABB Corporate Research, Baden, Switzerland}\\
       \affaddr{\{avin, zvilo, turkel\}@cse.bgu.ac.il, yvonne-anne.pignolet@ch.abb.com}}

\maketitle

\begin{abstract}
In many societies there is an \emph{elite}, a relatively \emph{small} group of \emph{powerful} individuals that is \emph{well connected} and highly \emph{influential}. 
Since the ancient days of Julius Caesar's senate of Rome to the recent days of celebrities on Twitter,  the size of the elite 
is a result of conflicting social forces competing to increase or decrease it.

The main contribution of this paper is the answer to the question how large the elite is at equilibrium.
We take an \emph{axiomatic approach} to solve this: assuming that an elite exists and it is \emph{influential}, \emph{stable} and either \emph{minimal} or \emph{dense},  we prove that its size must be $\Theta(\sqrt{m})$ (where $m$ is the number of edges in the network).


As an approximation for the elite, we then present an empirical study on nine large real-world networks of the subgraph formed by the highest degree nodes, also known as the rich-club.  Our findings indicate that elite properties such as \emph{disproportionate influence}, \emph{stability} and \emph{density} of  $\Theta(\sqrt{m})$-rich-clubs are universal properties and should join a growing list of common phenomena shared by social networks and complex systems such as ``small world,'' power law degree distributions, high clustering, etc.

\end{abstract}

\section{Introduction}

``How many scientists in your research area are in the group that \emph{controls} and \emph{influences} the progress of the field?," ``How many students are in the group 
that are responsible for the class \emph{spirit}?," ``How many representatives should a political party have in its decision making body?"
All these questions are about the \emph{elite} of a social network, namely about the \emph{size} of the subgroups that are, in some sense, controlling  the network. Can we have a \emph{rule of the thumb} to answer these questions raised for different networks? 
The current paper studies the structure of the elite and provides such a rule. But first, let us look at which other properties of social networks have been studied so far.

Notable examples for the basic \emph{universal properties} exhibited by complex systems and social networks are short average path lengths (a.k.a. the ``small world" phenomenon), high clustering coefficients, heavy-tailed degree distributions (i.e., scale-free networks), 
navigability, and more recently also dynamic properties such as densification and a shrinking diameter \cite{watts1998collective,albert2002statistical,leskovec2007graph,newman2010networks}. Traditionally, when ``new'' universal network properties are found by empirical measurements, a new random graph model has been proposed that generates networks which exhibit these properties. 
Therefore, the empirical findings led to the emergence of a variety of random graph models for evolution of social networks. 
And in turn, these evolutionary models have been used to predict and better understand the basic mechanisms that govern the behavior of social networks.
Some of the most popular random models are the {\em Barab{\'a}si-Albert Preferential Attachment} model~\cite{albert2002statistical}, the {\em Small-World} model~\cite{watts1998collective,kleinberg2000small}, the {\em Copy} model~\cite{kumar2000stochastic}, {\em Forest Fire} model~\cite{leskovec2007graph}, and more recently the {\em Affiliation Networks} model~\cite{lattanzi2009affiliation}.

  In the current work, we study 
 another basic and important phenomenon exhibited in the structure of
(social) networks: the existence of an \emph{elite} in the network.
In the Cambridge Dictionary the elite is defined as:
\begin{quote}
``\emph{The richest, most powerful, best educated or best trained group in a society.}''
\end{quote}

\noindent Other definitions (e.g., Wikipedia) emphasize in addition  that the elite group is ``small'' and ``well-connected.''

The first question we try to answer about the elite, is maybe the most basic one:
``What is the size of the elite, if it exists, in complex networks?'' Moreover, is the elite size a universal property of complex networks, similar to other universal properties we mentioned earlier? In particular, we investigate if the elite size can be expressed in relation to the total network size or other properties of the network. 

Instead of providing an evolutionary model, we take an \emph{axiomatic approach} 
to study these questions. 
The axiomatic approach has been used successfully in many fields of science, such as mathematics, physics, economy, sociology and computer science. See \cite{Andersen2008Trust-based,Geiger1991Axioms} for two examples in areas related to ours.
Perhaps the most well-known examples of this approach are Euclidean geometry and Newton's laws of motion.
In the early 20th century, the axiomatic approach was used successfully, e.g., by von Neumann in quantum physics and in utility theory (with Morgenstern),
and later in economy as well (e.g., by Nash, Vickrey, Aumann and Shapley).

Employing an axiomatic approach instead of providing a model has two main advantages for studying social and complex networks. 
While a basic random model provides us with a mechanism that generates networks with properties similar to the ones observed empirically,
 a suitable set of axioms attaches an ``interpretation'' or ``semantics'' to the observations.
Thus, a mechanism alone does not necessarily advance our understanding of the {\em meaning} of the phenomenon.  
In contrast when we describe a phenomenon in networks with a set of axioms, we can assess, to some degree, whether the axioms are ``believable''  
by how plausible their meaning is.
The second advantage is that once we agree on the axioms, it becomes possible to draw conclusions using logical arguments,
for example, inferring \emph{asymptotic} behavior, which is not always clear from empirical finding (i.e., since it is hard to determine the value of the constants involved). In our case, the axioms we propose allow us to answer the question of how large the elite is. 
As a consequence, the axiomatic approach is, in some sense, stronger than providing a particular model, since, if you agree on the axioms and their implications, 
\emph{every} model should be consistent with them. Furthermore, axioms can steer the search for new models in addition to empirical evidence.

Returning to elites, we observe that the size of the elite is determined by conflicting forces. Forces that try to increase  the elite size versus forces that try to reduce its size.  As a motivation for this tension and process we were inspired by historical anecdotes.
Ancient Roman history provides an example for forces that aspire to increase the size of the Roman elite, namely, the senate~\cite{humphries2003roman,kraus2010livy}.
When established around 750 BC, the senate included representatives of the first 100 families. When the population grew, it was extended, reaching 300 members in 509 BC, and then 600 senators in 80 BC. It was Julius Caesar who finally increased the senate to around a thousand senators. Interestingly, this number is about the square-root of the Rome population at the time, one million.
On the other hand, the French Revolution 1789--1799, provides an example for a society that reduced its elite size.
 During the Reign of Terror from 1793 to 1794 thousands of French elite members, the  nobility, were executed by the guillotine~\cite{greer1935incidence}.
 
Going back to our days, in the first part of this paper we make an effort to formalize some of these social forces as axioms, and to derive the elite size at equilibrium. Following the elite definitions mentioned earlier, our axioms try to capture, in a formal graph-theoretic way, that the elite is \emph{influential}, \emph{stable} and either of \emph{minimum-size} or \emph{dense}. The most important contribution of this paper is that these axioms imply an elite size of  $\Theta(\sqrt{m})$, where $m$ is the number of edges in the network.
Note, that in many networks $m$ is not known, but if we assume additionally that $n$ (the number of nodes in the network which is often easier to determine than the number of edges) and $m$ are of the 
same order of magnitude, i.e., the network is sparse, then our \emph{rule of thumb} for the size of the elite is $\Theta(\sqrt{n})$.

Intuitively, the best candidates for the elite group are nodes with the highest degree in the network, also called the \emph{rich-club} \cite{zhou2004rich}. These nodes are well-known to exist following the scale-free nature of complex systems and the power-law degree distribution that enforce such ``superstars''; nodes with degree well above the average degree (a.k.a. ``hubs'' and ``connectors'') \cite{gladwell2000tipping}.
Previous research on the rich-club phenomenon already demonstrated the existence of some interesting properties like the tendency of high degree nodes to be well connected among each other \cite{colizza2006detecting,mcauley2007rich,zhou2004rich}.
The importance of the rich-club with respect to the whole network was considered in~\cite{xu2010rich} which shows that the rich-club connectivity has a strong influence on the assortativity and transitivity of a network. 
Based on these findings, the rich-club can be seen as an approximation of the elite of a complex network due its influence on the rest of the network. We refer to Section~\ref{sec:relwork} for a discussion of related notions.

In the second part of the paper we perform an systematic empirical study on the structural properties of the $k$-\emph{rich-club subgraph}; the graph that is induced by the $k$ highest degree nodes. To this end, we analyze nine real-world complex networks. Our results indicate that the elite size, as well as others of its structural properties, are universal, as the measured quantities are similar on all the networks under scrutiny.

\subsubsection*{Summary of our Results}

\textbf{Elite Axioms and Size}:
We propose a set of axioms comprising the elite's influence, stability, density and compactness. From these axioms one can conclude that the size of the elite is in the order of $\sqrt{m}$, where $m$ is the number of edges in the network.

In contrast, previous research on the most influential group considered it to be a \emph{constant fraction} of the network , for example in the 80--20 rule of Pareto in economics \cite{pareto1971translation} in attacks to break networks~\cite{albert2000error},
or in rich-club studies \cite{zhou2004rich, mislove-2007-socialnetworks}.
The fact that the elite is of size $\Theta(\sqrt{m})$, an \emph{order of magnitude} less, can hence have very important consequences for several domains, as discussed in more detail in Section~\ref{sec:discussion}.\\

\noindent\textbf{Empirical Findings}:
We present measurements on nine existing social networks and complex systems for a variety of parameters for rich-clubs of growing size,  providing empirical evidence that for a  $\sqrt{m}$-rich-club, the following statements hold for its structure and its interaction with the whole network\footnote{Our findings on the density and influence reinforce, on additional real-world networks, some of the claims of previous rich-club studies. Our other findings reveal unknown features on the rich-club. In this paper we study the previously known and unknown properties together since they were never considered before in the context of a $\sqrt{m}$-elite size.}.


\emph{Inner Structure}:
(i) The induced subgraph of the $\sqrt{m}$-rich-club of existing social networks is \emph{dense}, in particular much denser than the whole network.
(ii) The largest connected component of this subgraph contains all but a handful (<33) of the rich-club nodes.
(iii) The average degree of the $\sqrt{m}$-rich-club in its induced subgraph is significantly higher than the average degree of the whole networks.
Note that these findings are \emph{not} a mere consequence of the fact that the rich-club contains the highest degree nodes (cf. to networks  with the same number of edges generated according to some complex network models discussed later).

\emph{Influence}:
The elite has a ``disproportionate'' power towards the society. In graph terms,
a significant constant fraction of nodes outside the $\sqrt{m}$-rich-club have a neighbor in the $\sqrt{m}$-rich-club. Related to this is the fact that the size of the cut between the $\sqrt{m}$-rich-club and the rest of the network is a significant constant fraction of all edges in the network.

\emph{Stability}:
The elite is stable against ``outside'' pressure from the society, i.e., there is a balance between the exposure to opinions inside and outside the elite. In graph terms, the ratio between the number of edges from the $\sqrt{m}$-rich-club to the remaining nodes (``outside'' pressure) and the number of internal edges (``inside'' pressure) of the $\sqrt{m}$-rich-club is constant.

\emph{Symmetry}:
In directed networks the $\sqrt{m}$-rich-club is significantly more symmetric than the whole network.

\cancel{
\textbf{Evolution}:
There is a high correlation between the high degree nodes and the \emph{seniority} of the network members. Note that while some models predict this well, others  do not.
}
Some of the above properties might have been known on an anecdotal level or may seem obvious; however, they have not been measured together for growing rich-clubs and they cannot be explained by \emph{only} considering the fact that the $k$-rich-club contains the highest degree nodes. It does not hold for arbitrary networks that the structure of the $k$-rich-club has these properties. In order to demonstrate this, we compare our findings on real-world data to the properties the popular Erd\"{o}s-Renyi model, the Barab\'{a}si-Albert model and the Affiliation networks model exhibit. 

In the next section we introduce the notions, data sets and models used in this paper, followed by a set of axioms and the derivation of the elite's size in Section~\ref{sec:size}.
After presenting our measurement results in more detail in Section~\ref{sec:measure} and reviewing related work in Section~\ref{sec:relwork}, we discuss our findings and some major open questions raised by them in Section~\ref{sec:discussion}.

\section{Datasets and Models}
Today several popular online social networking sites like Facebook, Twitter, Flickr,
YouTube, Orkut, and LiveJournal exist. These networking sites are based on an
explicit user graph to organize, locate, and share content as well as contacts.
In many of these sites, links between users are public and can be crawled
automatically. This allows researchers to capture and study a large fraction
of the user graph. The obtained data sets present an ideal opportunity to
measure and study online social networks at a large scale.
Mislove et al.~\cite{mislove-2007-socialnetworks,mislove-2008-flickr,viswanath-2009-activity} have collected data from the most prominent
online social networks and made them available to the research community. We
used their data on Facebook, Livejournal, Orkut, Flickr and YouTube
in addition to data provided by the Stanford Large Network Dataset Collection
(http://snap.stanford.edu/data/) on Autonomous Systems (AS) and Wikipedia link graphs. Furthermore, we
studied the rich-club of Twitter \cite{kwak2010twitter} and a citation network
(who cites whom) derived from DBLP and the ACM digital library.

To find out if the rich-club of real life complex networks is structurally different from arbitrary networks and to examine the rich-club of some well known graph models, we generated some graphs according to the
the Erd\"{o}s-Renyi random graph model, the Barabasi-Albert model and the Affiliation model.
One of the first and most simple models for networks is the Erd\"{o}s-Renyi (ER) random graph model \cite{erdos1960evolution}. In this model an edge between each pair of nodes exists with equal probability $p$, independently of the other edges.
One model to generate scale-free graphs exhibiting some properties found in real networks is the \emph{Barab\'{a}si-Albert} (BA) model \cite{albert2002statistical}. It captures growth and preferential attachment. More precisely it models the evolution of a social network, where nodes join the network and build links to existing nodes, based on their degree. The higher the degree of a node, the more likely it is to attract new nodes to connect to it. 
The network starts as an initial network of $m_0$ nodes.
New nodes are added to the network one at a time. Each new node is connected to $m'\leq m_0$ existing nodes with a probability that is proportional to the number of neighbors that the existing nodes already have. Formally, the probability $p_i$ that the new node is connected to node $i$ is \cite{albert2002statistical}
 $ p_i = {deg(i)}/{\sum_j deg(j)},$
where $deg(i)$ is the degree of node $i$.
In this work we adopt the convention $m_0=m'$ and start with an initial network forming a complete graph (clique).
%
Another model, based on a bipartite \emph{affiliation} graph from which a social network is derived, was presented in \cite{lattanzi2009affiliation}. The affiliation graph models the fact that people (``actors'') are typically connected to other people via ``societies'' (e.g., schools we visited, streets we live in, companies we work for, etc.). The social network is obtained by folding the bipartite graph, i.e., by generating  an (undirected) edge in the social network for paths of length two in the affiliation graph. The affiliation graph evolves by letting new actors and societies copy another node's neighbors with some probability in addition to preferential attachment edges based on the degree.
For each of these models we produced graphs with 1 million nodes. The parameters we used were $p = 0.00002$ for the ER model, $m' = 10$ for the BA model, and $c_q = c_u = 2$ (the number of edges added in 1 evolution step), $s = 2$ (the number of edges added by preferential attachment) and $\beta = 0.5$ (how often the left/right side of the bipartite graph grows).
We decided to use these models as most other models known to us are based on variations and combinations of these models.
All data sets (with degree rank as node identifiers) that we used in this paper are publicly available by emailing avin@cse.bgu.ac.il.

\section{Elite Size - Axiomatic Approach}\label{sec:size}

\begin{figure}
	\centering
	\includegraphics[width=0.8\columnwidth]{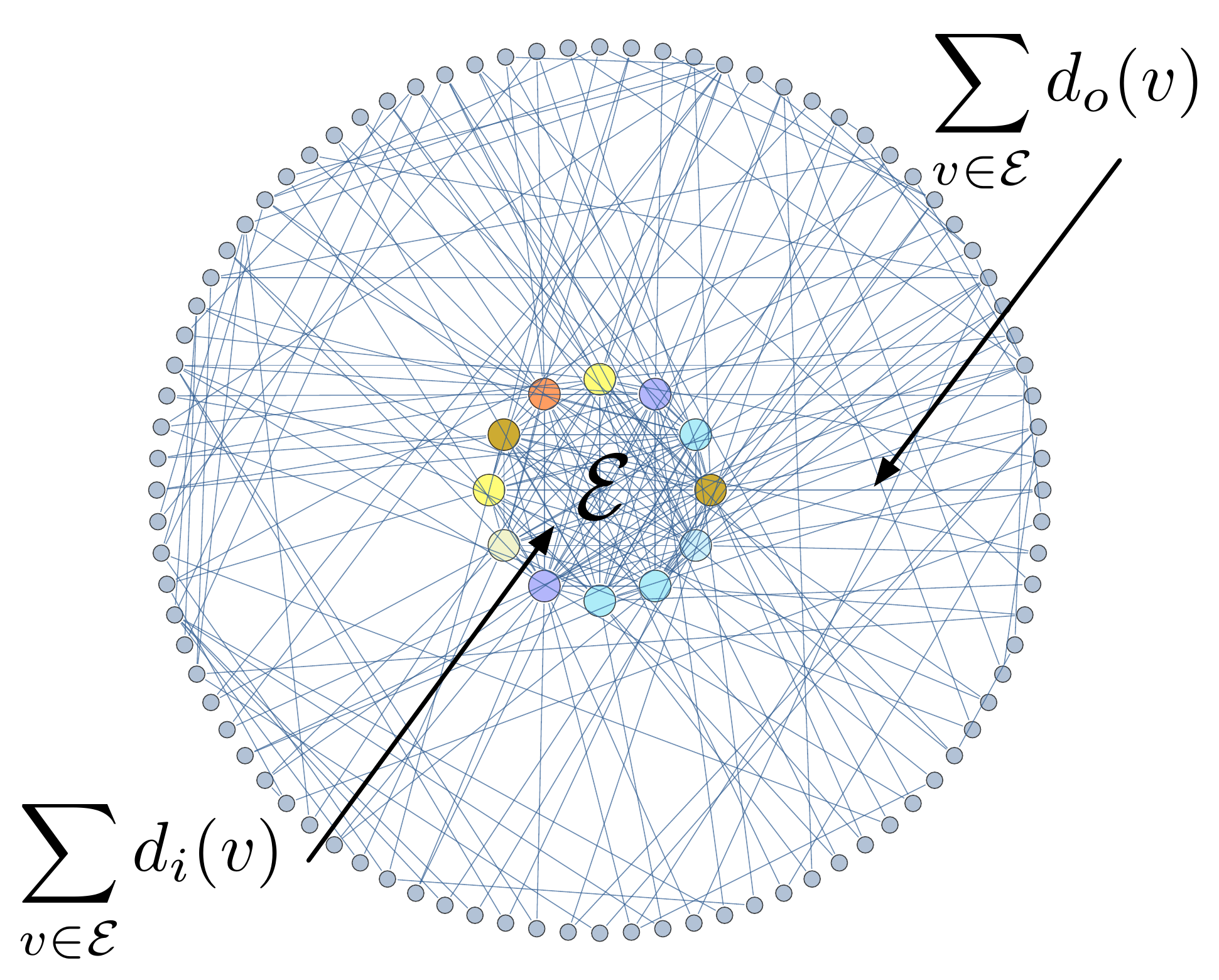}
	\caption{Graphical demonstration of the elite and parameters in the Axioms. $m$ is the total number of edges in the graph, $\E$ is the elite set (colored nodes),  $\sum_{v\in \E} d_o (v)$ is the total number of out-going edges from the elite and $\sum_{v\in \E} d_i (v)$ is the total number of edges within the elite.}
	\label{fig:combined_density}
\end{figure}

In this section we address one of the basic and most intriguing questions about the elite: what is its size? Many definitions of the elite indicate
that the elite is a \emph{small} group compared to the whole population. But what is the ``right'' size? What is small?

To answer these questions, and to explain our empirical results we take an axiomatic approach: we assume that an elite features some basic properties in order to maintain its power in the society and based on these we infer its size.
We claim that the elite must be socially \emph{stable} and \emph{influential}. 
 By adding either a \emph{density} or \emph{min-size} property,
we conclude that the elite size is in the order of $\Theta(\sqrt{m})$, where $m$ is the total number of edges. When the number of
edges is proportional to the number of nodes $n$, then the elite size turns out to be $\Theta(\sqrt{n})$.\footnote{While traditional research on social networks
assumed that $m=\Theta(n)$, more recent models and observations\cite{lattanzi2009affiliation,leskovec2007graph} show that $m=\omega(n)$, so $m$ and $n$ might differ in the order of magnitude.}

More formally, let the set $\E\subset V$ denote the elite consisting of $\abs{\E}$ nodes. For a node $v\in \E$ let $d_i(v)$ denote the internal degree of $v$ within $\E$, i.e., how many neighbors of $v$ belong to the set $\E$; analogously $d_o(v)$ denotes the number of neighbors of $v$ that are outside of the set $\E$, i.e., in $V\setminus\E$. An illustration of these notions can be found in Fig.~\ref{fig:combined_density}.
Let $c_1, c_2, c_3$ be some constants. We postulate the following axioms for the elite set $\E$.

\begin{table*}[htb!]
	\centering
		\begin{tabular}{|l| r | r |r|r|r|r|}
\hline
Data 		&	\multicolumn{1}{c|}{n} 	&	\multicolumn{1}{c|}{m}	&	 \multicolumn{1}{c|}{$\sqrt{m}$-rich-club}	&	 \multicolumn{1}{c|}{Influence - $c_1$}	&	\multicolumn{1}{c|}{Stability - $c_2$}	&	\multicolumn{1}{c|}{Density - $c_3$} \\ \hline
Youtube 	&	1138500 & 2989945 	& 1729 	& 35.60\% &	7.90\% & 5.70\%  \\ \hline
Facebook &	63732 &	817031 &	903 &	19.30\% & 	31.90\% &	12.40\%	\\ \hline
Livejournal& 	5204177 & 49163589 &	7011 &	9.50\% &	20.30\% &	3.90\%	\\ \hline
Orkut	& 	3072442 &	117174174 &	10824 &	10.20\% &	26.20\% &	5.40\%	\\ \hline
Flickr	& 	2302926 &	22830535 &	4778 &	34.80\% &	39.60\% &	27.50\%	\\ \hline
AuthorCitations &	85055 &	1234030 &	1110 &	31.80\% &	24.40\% &	15.60\%   \\ \hline
Wikipedia &	1870710 &	36473378 &	6039 &	41.80\% &	6.00\% &	5.00\%  \\ \hline
AS 		&	33560 &	75621 &	274 &	65.30\% &	16.80\% &	22.20\% \\ \hline \hline
Average 	 & & & &	  	  	31.04\% &	21.64\% &	12.21\% \\ \hline
STD 	  	& & & &  	  	18.37\% &	11.43\% &	8.90\%  \\ \hline \hline
ER Model	&1000000 &	9974503 &	3158 &	\cellcolor{gray!35}1.10\% &	\cellcolor{gray!35}0.50\% &	\cellcolor{gray!35}0.00\% \\ \hline
BA Model	&1000000 &	9973255 &	3158 &	11.20\%& 	5.30\% &\cellcolor{gray!35}	1.20\% \\ \hline
Affiliation Model & 	1000000 &	32092651 &	5665 &	11.60\% &\cellcolor{gray!35}	220.50\% & \cellcolor{gray!35}	51.30\% \\ \hline
\end{tabular}
\caption{Basic properties of the examined networks and models (\# of nodes, \# of edges) and the Axiom constants $c_1, c_2, c_3$ when the elite is the $\sqrt{m}$-rich-club. The average and standard deviation pertain to the real networks only. A gray cell indicates a large deviation of the model from the average of the real data.
}
\label{table:prop}
\end{table*}
\begin{itemize}
\item[1.]  \textbf{Influence.} The number of out-going edges from the elite is a constant fraction of the total number of edges in the network.
\begin{align}
\sum_{v\in \E} d_o (v)  \geq  c_1\cdot m  \label{eq:influence}
 \end{align}
  for $0 < c_1<1$.\\
 \emph{Motivation}: This axiom captures the power and influence that are associated with the elites. In complex networks an edge can be interpreted as a source of influence, thus a powerful group must control a large fraction of edges in the network.
\item[2.]  \textbf{Stability.} The number of edges within the elite is proportional to the number of out-going edges from the elite.
\begin{align}
\sum_{v\in \E} d_i(v)  \geq  c_2 \cdot \sum_{v\in \E} d_o(v) ,\label{eq:stability}
 \end{align}
 for $0 < c_2<1$.\\
 \emph{Motivation}: In order to adhere to its opinion, the elite must be able to resist ``outside'' pressure; otherwise individuals in the elite
 will change their option and will be influenced instead of being influential. Consider a node $v\in\E$ that makes decision based on a weighted majority  of her friends. Since people in the elite are, by definition,  more powerful (e.g., rich, educated, etc.), elite members' opinions count for more when $v$ consults it neighborhood. If we weigh friends within the elite with a power of 1, then the weight of the outside friends $w$ will be less than 1.
 $c_2$ represents the (average) power we associate with friends outside the elite in a stable case. Therefore, we expect that $c_2 < 1$ in real networks.


\item[3.]  \textbf{Minimum-Size/Compactness.} The number of elite members tends to be as small as possible. \\
\emph{Motivation}: This axiom is based on some basic principals in science like \emph{Principle of minimum energy}, \emph{Principle of least action} and \emph{Occam's razor}. Given that other things are equal (such as influence and stability), the elite size will tend to be as small as possible.
In social terms this can be motivated by the idea that if the elite holds a given revenue or power, it will attempt to split it between as few members as possible.

\item[4.]  \textbf{Density.} The elite is dense
\begin{align}
\sum_{v\in \E} d_i(v) \geq  c_3 \cdot {\abs{\E} \choose 2} \label{eq:density}
 \end{align}
 for $0 < c_3<1$.\\
\emph{Motivation}: The goal of the density property is to capture the idea that the elite is a social ``clique'' where ``everyone knows everyone.''
The density property holds when each member of the elite knows (on average) a constant fraction of the elite members.
\end{itemize}
Based on these axioms, we can infer the size of the elite $\abs{\E}$. First we show a lower bound.
\begin{claim}\label{clm:lower}
If the elite satisfies Axioms 1 and 2 the size of the elite is: $\abs{\E} = \Omega(\sqrt{m})$.
\end{claim}

\begin{proof}
First note that $\abs{\E}^2 > \sum_{v\in \E} d_i(v)$. Using Eq. \eqref{eq:influence} and  \eqref{eq:stability} we get
$\abs{\E}^2 > \sum_{v\in \E} d_i \geq c_1 \cdot c_2\cdot m$ 
which implies
$\abs{\E} \ge \Omega(\sqrt{m}).$
\end{proof}

It is important to note that Axioms 1 and 2 alone do not guarantee a \emph{small} elite. Take for example a linear size elite, i.e., $\abs{\E}=\Theta(m)$, of constant degree, e.g., a constant degree expander \cite{hoory2006expander}. If additionally each member of the elite is connected to a constant number of nodes outside of the elite, the resulting elite is both stable and influential. In order to derive a small elite size, we must assume additional axiom(s). We next show that assuming either Axiom 3 or 4 enables us to conclude an elite size in the order of $\sqrt{m}$.

\begin{theorem}
If the elite satisfies Axioms 1, 2 and 3 ,the size of the elite is: $\abs{\E} = \Theta(\sqrt{m})$ and the elite is dense. 
\end{theorem}
\begin{proof}
The upper bound of  $\abs{\E} =O(\sqrt{m})$ follows directly from Claim \ref{clm:lower} and Axiom 3. Now assume the elite is not dense.
Then $\sum_{v\in \E} d_i = o(\abs{\E}^2) = o(m)$. But this contradicts $\sum_{v\in \E} d_i \geq c_1 \cdot c_2\cdot m$; thus, the elite must be dense.
\end{proof}

\begin{theorem}
If the elite satisfies Axioms 1, 2 and 4, then the size of the elite is: $\abs{\E} = \Theta(\sqrt{m})$ and the elite is compact. 
\end{theorem}
\begin{proof}
As before, $\abs{\E}^2 > \sum_{v\in \E} d_i(v)$ and using Eq. \eqref{eq:influence},  \eqref{eq:stability} and \eqref{eq:density}, we get
$\abs{\E}^2 > c'\cdot m \ge c'' \cdot \abs{\E}^2$ for some constants $c', c''$. Hence, it must hold that $\abs{\E} = \Theta(\sqrt{m})$ which means the elite is compact (i.e., of minimum possible size when assuming Axiom 1 and 2 ).
\end{proof}



An important note is in place: the Axioms and results above hold for \textbf{undirected} networks only.  For \textbf{directed} networks like Twitter, a more general treatment is needed and we leave this for future work.

In context of the above results, we would like to mention the pioneering work of Linial et. al.~\cite{linial1993sphere} and Peleg ~\cite{peleg2002local} where a notion related to elites is discussed, namely coalitions. Coalitions are subsets of nodes that dominate neighborhood majority votings, they are also called monopolies. Among the results regarding monopolies are lower bounds for the size of these coalitions, namely $\Omega(\sqrt{m})$ nodes are necessary to control the outcome of all neighborhood majority votings of general graphs and the number of edges within the monopoly has to be in the order of $m$. Interestingly, we use a different perspective, but our Axioms 1 and 2 lead to the same result for a powerful elite. 

In the next section we present measurements to support our claim that the elite size in social networks is of size $\Theta(\sqrt{m})$.
We show that for existing networks a rich-club of size  $\Theta(\sqrt{m})$ is both stable and influential, where this is not the case for standard models like the BA and Affiliation model.  As discussed earlier we will focus on the rich-club of different sizes as an approximation for the elite of the network.

\section{Empirical Results}\label{sec:measure}
We studied a variety of parameters for rich-clubs of growing size for nine real networks and three theoretical models (Erd\"{o}s-Renyi (ER) model, Barab\'{a}si-Albert (BA) model and the Affiliation model). We paid most attention to $k$-rich-clubs for $k$ in the order of $\sqrt{m}$ to corroborate our hypothesis that rich-clubs around this size satisfy our axioms.
First we examined the constants $c_1, c_2, c_3$ relevant to our axioms properties of influence, stability and density, and then other relevant properties.

\begin{figure}
	\centering		
	\includegraphics[width=\columnwidth]{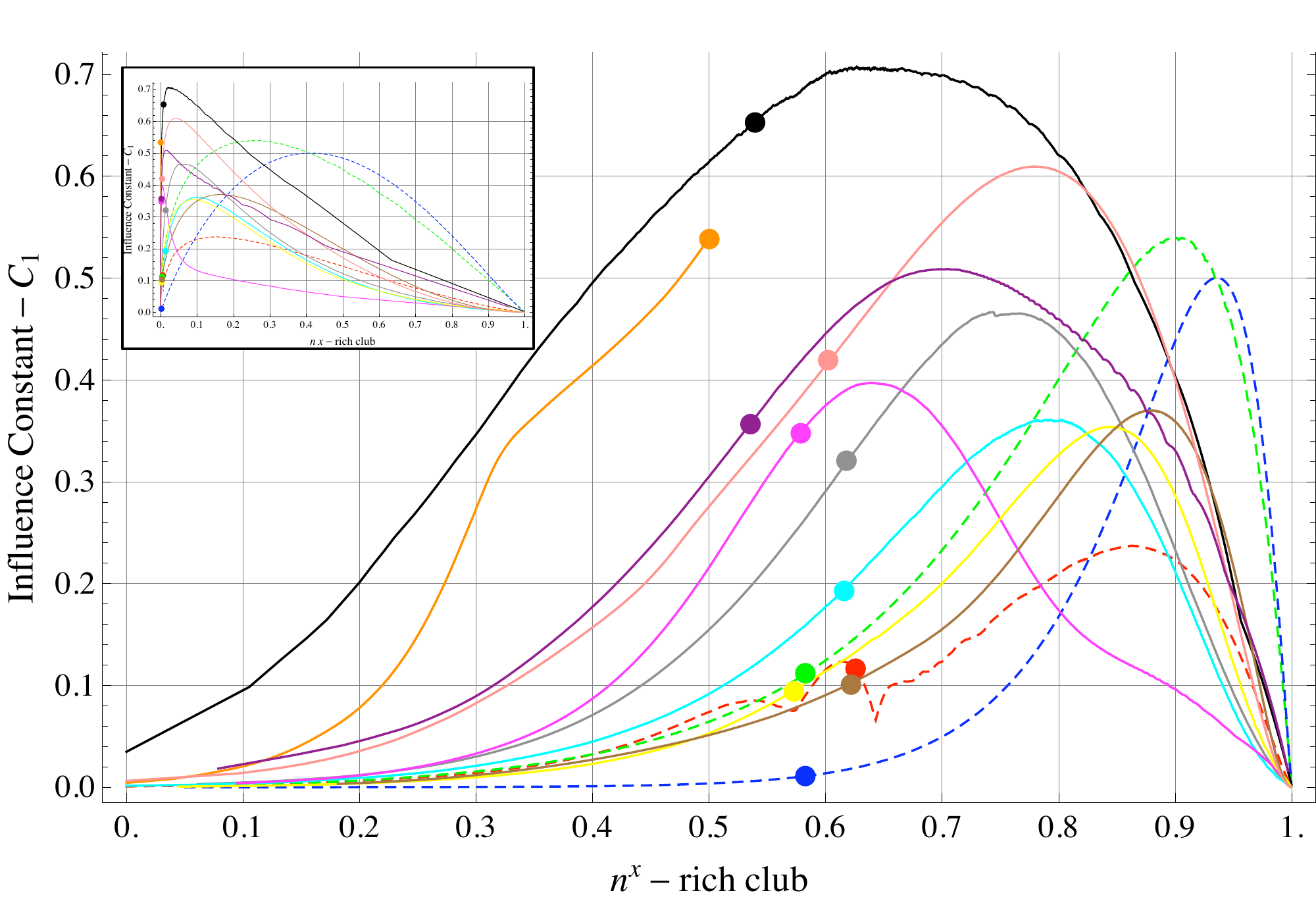}
	\footnotesize{(A)}	
	\includegraphics[width=\columnwidth]{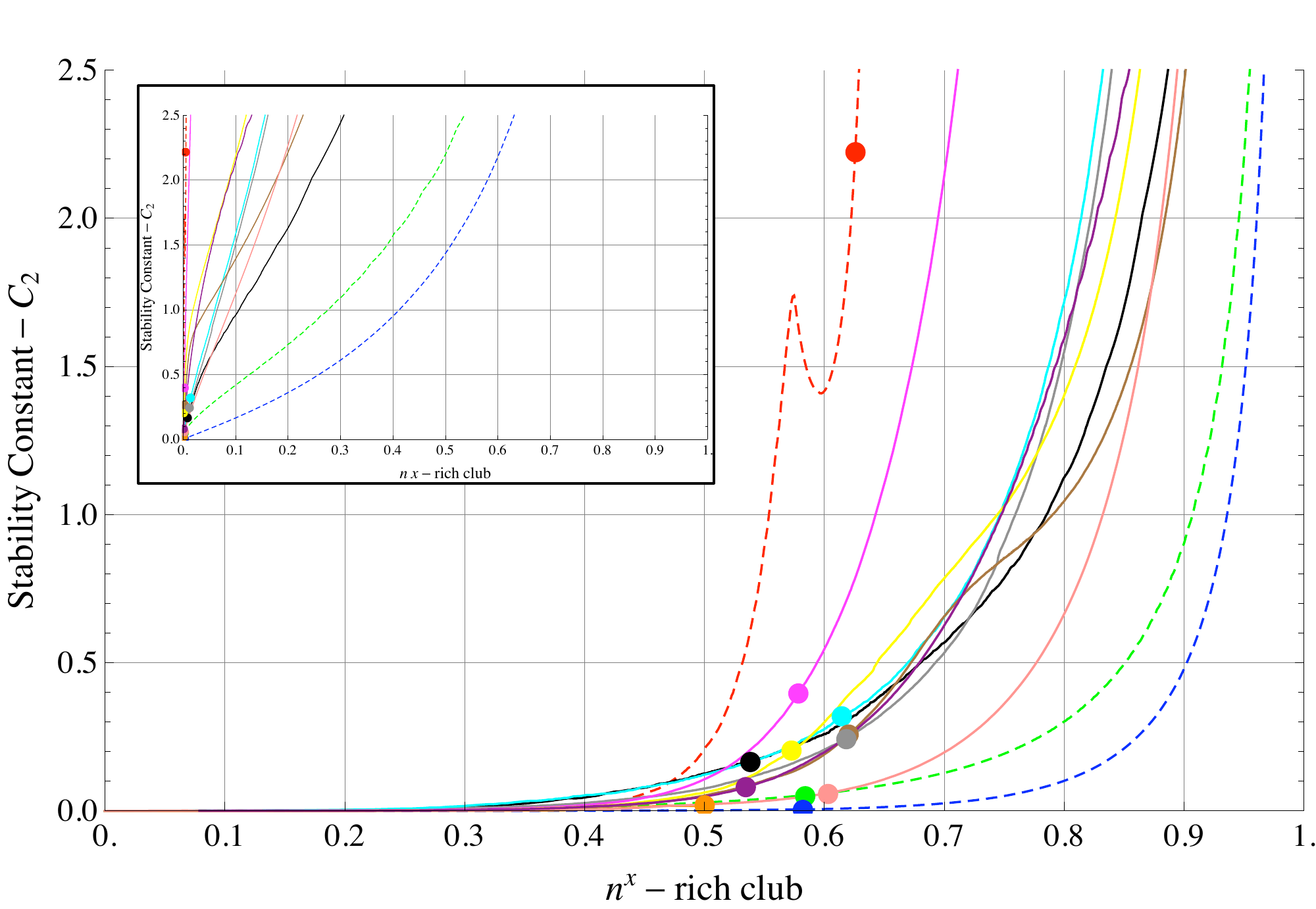}
	\footnotesize{(B)}	
	\includegraphics[width=\columnwidth]{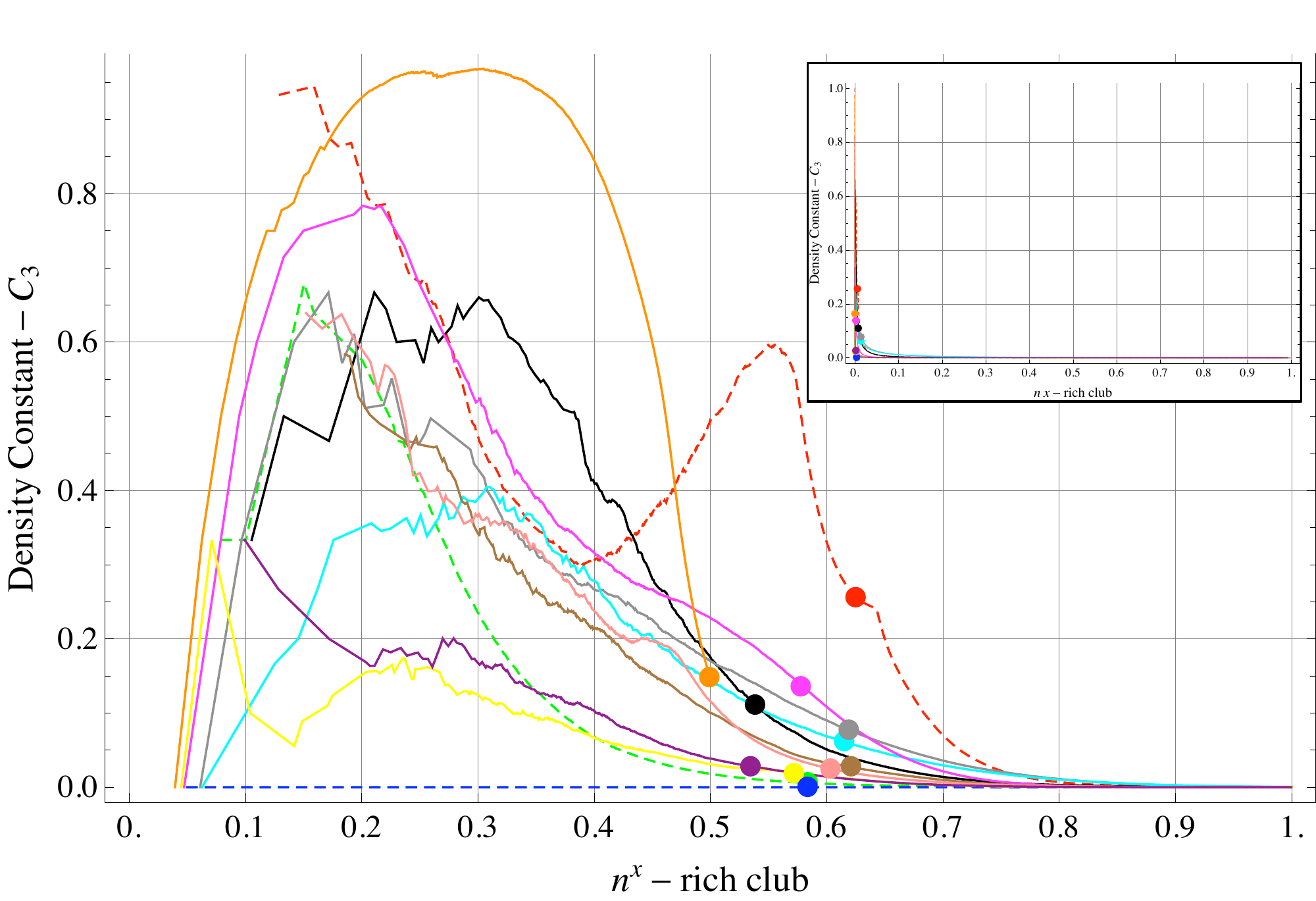}
	\footnotesize{(C)}	
	\includegraphics[width=\columnwidth]{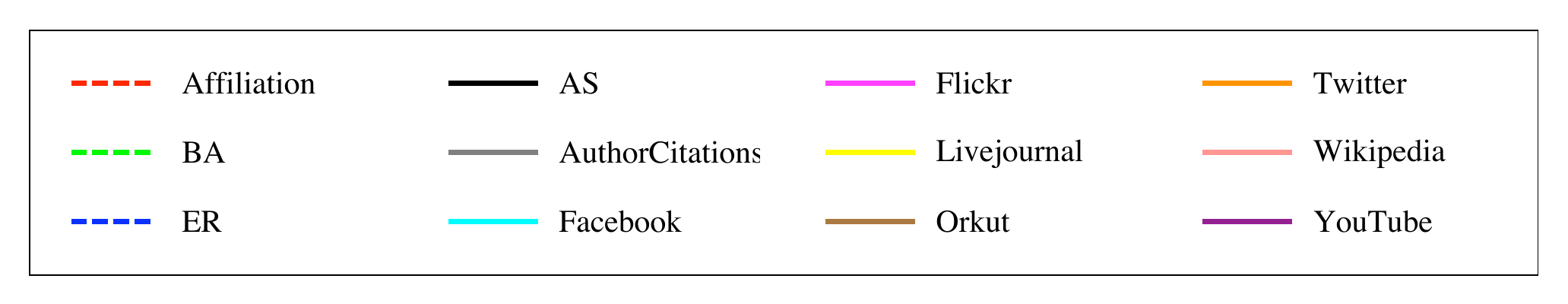}
	\caption{Three graphs that show the values for influence, stability and density from the axioms on nine real networks (solid lines) and three models (dashed lines) for $n^x$-rich-clubs. (A) the influence constant - $c_1$, (B) the stability constant - $c_2$ and (C) the density constant - $c_3$. The dots indicate where the $\sqrt{m}$-rich-club is located, and at $x=0.5$ the value for the $\sqrt{n}$-rich-club is depicted. The small figures are the same but with a linear scale, i.e., $k=nx$.}
	\label{fig:crossing}
\end{figure}

\subsection{Axiom constants - $c_1, c_2, c_3$}
Table~\ref{table:prop} gives a summary on basic properties of the networks under scrutiny
such as the number of nodes and edges ($n$ and $m$ respectively) and the influence, stability and density constants ($c_1, c_2, c_3$, respectively) for the $\sqrt{m}$-rich-club of each network. Gray cells indicate a large deviation by a model.
i) Influence - the influence constant fluctuates between $9\%-66\%$ in the real networks with high variance, $c_1$ of the BA and Affiliation models both fall in this range, but clearly (and intuitively) in the ER model, the  $\sqrt{m}$-rich-club is not influential enough.
ii) Stability - All the constants for the real networks are below one as expected. For the Affiliation model $c_2 > 2, (220\%)$, which contradicts our Axiom.
iii) Density - Real networks show a density constant between $4\%$ and $28\%$, while ER and BA are more sparse and the Affiliation model is more dense. We discuss these findings in more detail later.

We now take a broader view of the results and check the above constants for an increasing size of $k$-rich-club from $k=1$
to $k=n$. In order to compare networks of different sizes, we used plots where the $x$ axis (the rich-club size, $k$) is normalized to $[0,1]$. To focus on small $k$-rich-clubs, the $x$-axis describes the rich-club size for growing roots of the network size $n$, i.e., at $x \in[0,1]$ the measurement point for the $n^x$-rich-club, in particular at $x=0.5$ we have the values for the $\sqrt{n}$-rich-club. We emphasize the size of the $\sqrt{m}$-rich-club by adding a large dot at its location which differs for each network. In other figures (usually in a smaller size) we present the result on a linear $x$-axis, where $k=nx$ which demonstrates that interesting phenomena occur for a very small $k$ compared to the network size (i.e.,  in the order of $\sqrt{m}$).

Fig. \ref{fig:crossing} (A), (B) and (C) presents the results for influence, stability and density of growing rich-clubs, respectively. 
Regarding influence - $c_1$, we observe some similar characteristics for all networks:  i) influence increases monotonically until the maximum influence which is achieved at a rich-club size much larger than $\sqrt{m}$. ii) the constant $c_1$ is bound away from 0 at $\sqrt{m}$ (except in the ER model). Moreover, influence in most real networks is larger than in the models for low order $k$. Two extremes cases are the  AS (Internet routers) network and Twitter. In particular Twitter is a directed network where the $k$-rich-club seems to have a much larger (directed) influence. Since we have data of Twitter only for a $k$-rich-club up to $k=\sqrt{n}$, not the whole scale is presented for this network. 

For the stability constant $c_2$ we also observe similar behavior among the networks: stability monotonically increases with the rich-club size. Except for Twitter and the ER model, the constant is clearly bounded away from $0$ for $\sqrt{m}$-rich-club. 

Recall that our axiom results about the constants hold (and are well-defined) for the undirected case only, whereas Twitter is a highly asymmetric directed network. 
As noted when discussing Table~\ref{table:prop}, $c_2$ of the Affiliation model exceeds 1 for the $\sqrt{m}$-rich-club (and much earlier), which 
contradicts the second axiom. 

An important observation is that to increase both influence and stability, a larger $k$-rich-club is better. In contrast this is not the case with density. 
The density of the $k$-rich-club exhibits the opposite behavior (common to all networks, but the ER model): the maximum density is achieved 
at a rich-club size significantly smaller than $\sqrt{n}$, and from there on the density decreases monotonically (except for the Affiliation model). 
So for the $k$-rich-club to be denser, it must be smaller, while to be influential and stable its size needs to be larger.

We conjecture that these conflicting forces determine the ``right'' size of the elite.
The above empirical results (strengthened by the axioms) indicate that the balance between these forces, or the equilibrium, is achieved when the 
rich-club is in the order of $\sqrt{m}$. The graphs generated by the three theoretical models do not exhibit all of the properties (influence, stability and density) in the right scales as shown by the real networks.

To make the statements about the models more formal, we establish the following.
\begin{proposition}[ER model rich-club density] \label{thm:ER} $\\$
The expected number of edges in the $k$-rich-club of a ER graph is $o(k^2)$ for $p=o(1)$ and large enough $k$.
\end{proposition}

It is also not hard to show that the BA model does not fulfill the density requirement.

\begin{theorem}[BA model rich-club density]\label{thm:ba_rich-club_density} $\\$
The expected number of edges in the $k$-rich-club of a Barabasi-Albert graph is linear, i.e., $O(k)$.
\end{theorem}

\begin{proof}
No matter which nodes belong to the $k$-rich-club, each node has $m$ out-going edges
in the BA model. Hence the total number of edges within the rich-club cannot exceed
$2k\cdot m$. As a consequence, the rich-club is not dense if $k$ is not in the same order
of magnitude as $m$.
\end{proof}

A more challenging task from a theoretical point is to prove that the Affiliation model does not capture
the properties of the rich-club  of complex systems or social networks.
We leave this as a conjecture in this work and provide more empirical evidence in the next section.

\subsection{Maximum Sociability}
Another measure for the structure and connectivity of the $k$-rich-club is the \emph{sociability}.
The sociability of a graph is its normalized average degree, i.e., for $k$-rich-club with $m_k$ edges among the rich-club members this value is $$\frac{m_k/k }{\max_{1\leq k'\leq n}{m_k'/k'}}.$$ For a graph of growing size the maximum sociability captures
the size of the network at which its members are, on average, most socially involved (or influenced) in the community.
As mentioned earlier the average degree of the BA model is the same for any $k$-rich-club and therefore its sociability level is more or less constant after 5\% of the network size.
In contrast, real social networks are significantly different with the maximum sociability achieved at a $k$-rich-club of size around $n^{0.6}$.
This can be seen in Fig.~\ref{fig:sociability}, where the figure shows that the maximum is achieved at a small scale $k$-rich-club. Interestingly, all real social networks have a single peak for the maximum, this may indicate that this point is a good candidate to define the "right" size of the rich-club.
An exception to this rule is the Wikipedia graph with two maxima.
When examining Fig.~\ref{fig:sociability}, we notice that the maxima occur before or after $k=\sqrt{m}$ in some networks. One possible explanation is that our data sets are not complete, i.e. some nodes and edges are missing; another is that these networks are not in a balanced state, i.e., the elite will grow or shrink until an equilibrium is reached.

\begin{figure}[t!]
	\centering			\includegraphics[width=\columnwidth]{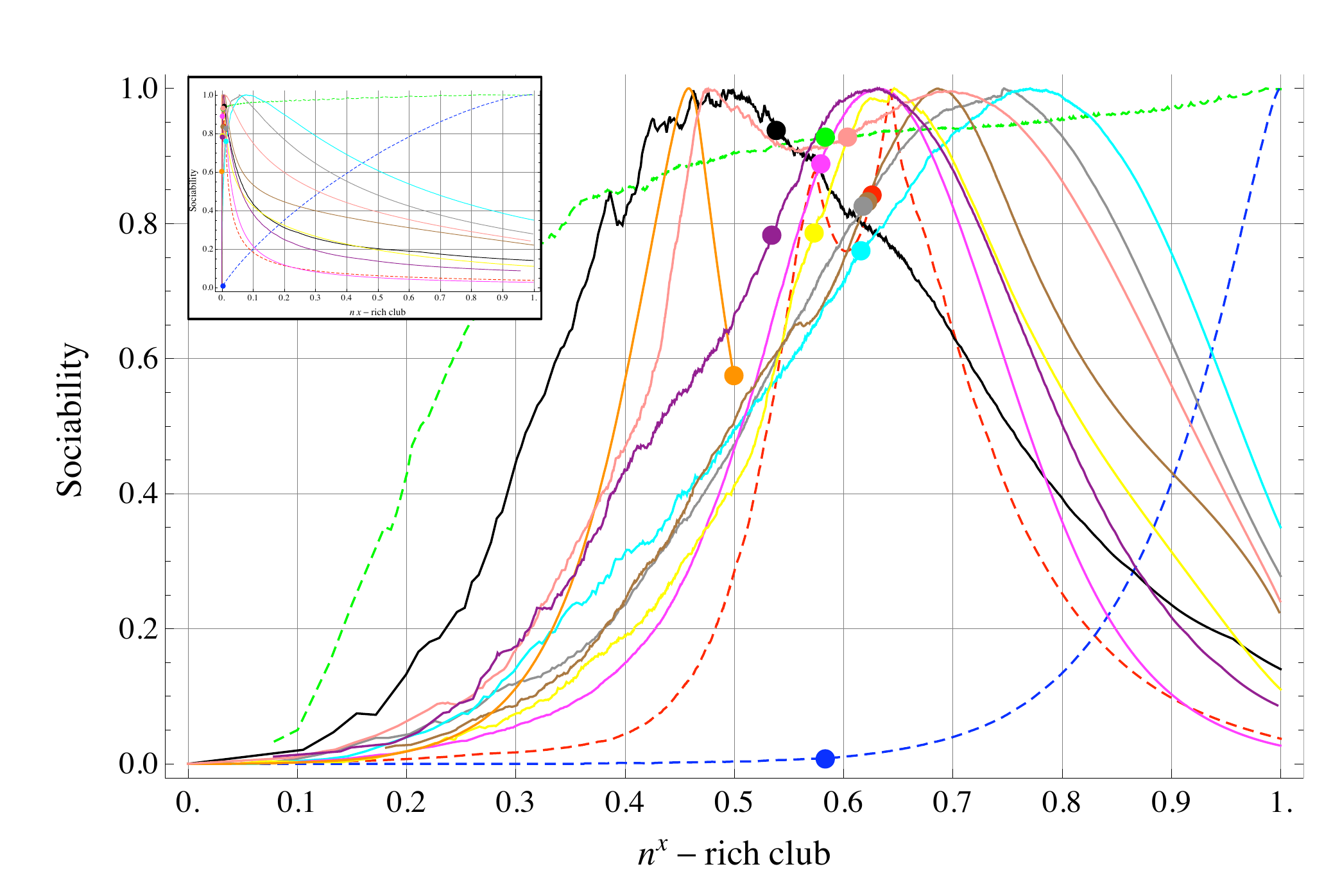}
	\caption{Maximum Sociability: This graph depicts the number of rich-club edges divided by the number of rich-club nodes with the maximum normalized to one. For a $k$-rich-club with $m_k$ edges, this value is $\frac{m_k/k }{\max_{1\leq k'\leq n}{m_k'/k'}}$. This ratio is equal to the average degree of the rich-club nodes.}
	\label{fig:sociability}
\end{figure}

\subsection{Elite Connectivity}
In social networks, the largest connected component (LCC) typically covers almost all nodes of the network. However, this does not imply that for any graph with a large LCC, it must hold that the LCC of the $k$-rich-club contains almost all $k$ nodes. E.g., we found when analyzing the size of the LCC of the $\sqrt{m}$-rich-club reveals that almost all nodes in the rich-club of the social networks belong to the LCC. The same holds for the BA and Affiliation model.
In the ER graph, however, most rich-club nodes do not have any edges to other rich-club nodes; hence, the rich-club is split into many separate components, most of them consisting of one node only.
\begin{table}[h]
\centering
		\begin{tabular}{|l|r|r|r|}
\hline
\multicolumn{1}{|c|}{data}	& \multicolumn{1}{c|}{$\sqrt{m}$-rich-club}  &	 \multicolumn{1}{c|}{\# comp}	&	 \multicolumn{1}{c|}{LCC} \\      \hline	
YouTube&	1729	&	9	&	1721	\\	\hline
Facebook&	903	&	1	&	903	\\	\hline
LiveJournal&	7011	&	16	&	6978		\\	\hline
Orkut&	10824	&	13	&	10812	\\	\hline
Flickr&	4778	&	1	&	4778	\\	\hline
Author Citations&	1110	&	1	&	1110		\\	\hline
Wikipedia&	6039	&	2	&	6038		\\	\hline
AS&	274	&	4	&	271	\\	\hline
ER &	3158	&	2888	&	3		\\	\hline
BA&	3158	&	1	&	3158	\\	\hline
Affiliation&	5665	&	1	&	5665	\\	\hline
\end{tabular}
\caption{Connectivity table of $\sqrt{m}$-rich-club. This table summarizes the number of connected components the $\sqrt{m}$-rich-club and the size of its largest connected component (LCC).}
\end{table}

\subsection{Symmetry}
In some networks the existence of an edge describes a reciprocal, symmetric relation between the two nodes involved (undirected network), whereas in other networks an edge from node $a$ to node $b$ (directed network) means that $a$ has a certain relationship with $b$ but not necessarily the other way around.
Classically, sociologists make a distinction between directed networks and undirected networks when analyzing them. E.g., the first question of a decision tree for the analysis of cohesive subgroups on page 78 of \cite{Wouter-05} is ``Is the network directed?'' The mathematical tools that are used differ depending on the answer, e.g., the notion of prestige (in-degree) does only apply to directed networks. On the other hand, in undirected graphs, degree centrality is used (see \cite{Wasserman-Faust-94}, Chapter 5). Clearly the directed graph model contains more information than its equivalent undirected version. However, in many networks it is impossible or difficult to derive who initiated a relationship and/or what the direction of an edge is.
For directed networks a natural question is whether the rich-club of the network is more symmetric than the rest of the network. Of our datasets the networks  Wikipedia, Flickr, YouTube, Twitter and ER graph are directed. The average symmetric degree in the rich-club has a unique maximum in all three real networks. The maximum ``ordinary'' average degree of the rich-club is reached slightly after the maximum of the symmetric rich-club degree.
Furthermore, it holds that the rich-club of the networks is more symmetric: the ratio between symmetric edges and all edges in the $k$-rich-club starts at almost 1 for $k=2$ and then decreases rather quickly until reaching almost zero when $k$ approaches $n$. At around the maximum sociability ($k\approx\sqrt{n}$) the symmetric edges are still a significant fraction of all edges.
In the ER graph model there are no symmetric edges which is not surprising for the chosen edge probability. Since the BA model and the Affiliation model are undirected, they cannot help to explain or model the high symmetry within the rich-club.
\cancel
{
\begin{figure}	\label{fig:symmetry}
	\centering \sffamily
	Wikipedia\\
		\includegraphics[width=\columnwidth]{wikipedia_growth_sym.pdf}\\
	Flickr\\
		\includegraphics[width=\columnwidth]{Flickr-links_sym.pdf}\\
	YouTube\\
		\includegraphics[width=\columnwidth]{YouTube_growth_sym.pdf}\\
	\caption{The symmetry level varies among different networks, as depicted in these graphs of Wikipedia, Flickr and YouTube (from top to bottom). The measurements have been normalized to have there maximum at one. The dark line shows the ratio between the number of symmetric edges in the rich-club and the number of rich-club nodes, i.e., the average symmetric degree in the rich-club. As a comparison the dashed line shows the number of all edges in the rich-club divided by the number of nodes in the rich-club, i.e., the average degree (sociability).  The dotted line depicts the ratio between the number of symmetric edges and the total number of edges in the rich-club.}
\end{figure}
}

In addition we counted the number of symmetric edges in the $\sqrt{n}$-rich-club of Twitter. In the following table we can see that 89\% of the edges in the Twitter $\sqrt{n}$-rich-club are reciprocal, while in the whole Twitter network 22.1\% of all edges are reciprocal~\cite{kwak2010twitter}. \\ \\
\begin{small}
\begin{tabular}{|l|r|r|r|r|r|}
\hline
 \multicolumn{1}{|c|}{$m_{rc}$} &  \multicolumn{1}{c|}{total} &  \multicolumn{1}{c|}{min}  &  \multicolumn{1}{c|}{max} &   \multicolumn{1}{c|}{median}   &  \multicolumn{1}{c|}{avg}   \\ \hline
 directed &5,537,573 &0   & 3,778   & 656  & 852.07    \\ \hline
reciprocal &  4,952,210& 0   & 3,238  &  512   &762.00   \\ \hline
\end{tabular}
\end{small}
\\

When considering Twitter we notice that the $\sqrt{n}$-rich-club features especially high symmetry. One possible explanation for this is that the rich-club of Twitter is much larger than in the other networks and that this increases the social pressure on each of its members to increase the symmetry. Another explanation is that for Twitter many tools exist that help Twitter users to organize their tweets, followers and the users they are following. Among other features, some of these tools offer the functionality to add a new follower to the list of people they are following. Presumably many of the high degree Twitter users apply such a software and ``follow back'' their followers automatically.
In order to find out if one of these theses is true, it is necessary to scrutinize data of other large networks and observe how the symmetry percentage changes with growing network size.

\section{Related Work}\label{sec:relwork}
One of the first papers about the fact that the highest degree nodes are
well connected examined the Autonomous Systems network~\cite{zhou2004rich} and coined the term rich-club coefficient for the ratio comparing the number of edges between nodes of  degree greater than $k$ to the possible number of edges between these nodes.  Colizza et al.~\cite{colizza2006detecting} refined this notion to account for the fact that higher degree nodes have a higher probability to share an edge  than lower degree vertices. They suggest to use baseline networks to avoid a false identification of a rich-club. More precisely they propose to use the  rich-club coefficient of random uncorrelated networks and/or the rich-club  coefficient of network derived by random rewiring of edges, while maintaining the degree distribution of the network. Weighted versions of the rich-club coefficient have been studied
in~\cite{opsahl2008prominence,serrano2008rich,zlatic2009rich}
The question how the rich-club phenomenon manifests across hierarchies is studied in~\cite{mcauley2007rich}.

As identifying the most influential nodes in a network is crucial to understand its members behaviour, many other articles considered a variety of notions related to the elite and/or the rich-club.
Mislove et al.~\cite{mislove-2007-socialnetworks} defined the \emph{core} of a network to be any (minimal) set of
nodes that satisfies two properties: First, the core must be
necessary for the connectivity of the network (i.e., removing
the core breaks the remainder of the nodes into many small,
disconnected clusters). Second, the core must be strongly
connected with a relatively small diameter.
As a consequence, a core is a small group of well-connected group of nodes that is necessary to keep the remainder of the network connected.
Mislove et al. used an approximation technique previously used in Web graph analysis,
removing increasing numbers of the highest
degree nodes and analyzed the connectivity of the remaining
graph. The core is thus the largest remaining strongly connected component.
They observed that within these cores the
path lengths increase  with the size of the core when progressively including nodes ordered
inversely by their degree. The graphs they studied in \cite{mislove-2007-socialnetworks} have a densely connected core
comprising of between 1\% and 10\% of the highest degree
nodes, such that removing this core completely disconnects
the graph.

Another definition for a core can be found in \cite{borgatti2000models}. Borgatti and Everett measured how close the adjaceny matrix of a graph is to the  block matrix  $\{\{1,1\},\{1,0\}\}$. This captures the intuitive conception that social networks have a dense, cohesive core and a sparse, unconnected periphery. Core/periphery networks revolve around a set of central nodes which are well-connected with each other, and also with the periphery. Peripheral nodes, in contrast, are connected to the core, but not to each other. On the other hand, there are "clumpy" networks consisting of two or more subgroups that are well-connected within each group but weakly connected across groups -- like a collection of islands.
When comparing networks with the same density, core/periphery networks have shorter average path lengths than clumpy networks.
In addition to formalizing these intuitions, Borgatti and Everett devised algorithms for detecting core/periphery structures, along with statistical tests for testing a priori hypotheses~\cite{borgatti2002ucinet}.

The nestedness of a network represents the likelihood of a node to be connected to the neighbors of higher degree nodes. When examining this property, block modeling of adjacency matrices arranged by the degree of the nodes is also used. E.g., Lee et al~\cite{lee2012scaling} studied such block diagrams for complex network models, and defined a simple nestedness measure for unipartite and bipartite networks to capture the degree to which different groups in networks interact.

Apart from analyzing the most influential nodes, many articles have studied a wide range of properties of social networks. E.g., the networks created by YouTube, Flickr, Facebook, Wikipedia and LiveJournal have been analyzed in depth in \cite{mislove-2007-socialnetworks, mislove-2008-flickr,viswanath-2009-activity}. Twitter has been studied for its applicability to spot trends, homophily and rumour spreading \cite{kwak2010twitter,jansen2009Twitter}. 
In addition there is a large body of papers studying the evolution of social networks \cite{albert2002statistical,leskovec2008microscopic,kumar2010structure,viswanath-2009-activity}, 
information dissemination and path lengths \cite{albert2002statistical,kempe2003maximizing, leskovec2008planetary,goel2009social}, and community structure \cite{leskovec2008statistical}, to name but a few examples.

\section{Discussion and Open Problems}\label{sec:discussion}
In order to make a step forward towards finding the ``correct'' size of the rich-club, we used rich-club expansion to determine a subset which
exhibits significant structural difference and influence to the rest of the network. This is closely related to the question of finding the elite of a complex network. Based on three axioms and on measurements performed on nine real-world networks, we conclude that the elite of a network consists of around $\sqrt{m}$ nodes and the $\sqrt{m}$-rich-club serves as a good approximation for the elite.

Our results do not only advance the theoretical understanding of the elite of social structures, but may also help to organise institutions better or to identify sources of power in social networks.

Reinforcing the claims of previous work on high degree nodes, our data analysis shows that many complex networks have a small subgraph which is much more dense than their complete network. In addition the structure of the whole network is influenced by this rich-club. This can be exploited to find good candidate networks for the  problem of finding the most dense subgraph (a NP-hard problem \cite{KortsarzP93} on general graphs). One can apply the following procedure: Sort the nodes according to their degrees and choose the most dense subgraph among the subgraphs that contain the first $k$ highest degree nodes. We hope that this heuristic can be turned into an approximation algorithm once there are better models that capture rich-club properties of complex networks.

In addition to the above we provide answers to the central question of how symmetry is spread among the edges of directed social networks. We show that edges inside the rich-club are much more symmetric than random edges that are not inside the rich-club. 

Apart from the size of the elite, an intruiging endeavour is to study the evolution of the elite. Interesting future research questions include the following: Are there universal rules for elite dynamics? Are early members more likely to remain in the elite if the network grows? Are there strategies for late comers to join the elite?



\section{Acknowledgements}
We would like to thank the anonymous reviewers and arXiv readers for their suggestions for improvements and pointers to related work.

\bibliographystyle{abbrv}
\bibliography{sigproc}
\end{document}